\documentclass[10pt]{llncs}
\def\nottoobig#1{{\hbox{$\left#1\vcenter to1.111\ht\strutbox{}\right.\n@space$}}}

\if01

\newtheorem{theorem}{Theorem}[section]

\newtheorem{lemma}[theorem]{Lemma}

\newtheorem{proposition}[theorem]{Proposition}

\newtheorem{definition}[theorem]{Definition}

\fi
\usepackage{epsfig}
\usepackage{amsmath}
\usepackage{amsfonts}

\newcommand{\calA}{{\cal A}}

\newcommand{\nat}{{\mathbb N}}

\newcommand{\poly}{{\rm poly}}

\def\nottoobig#1{{\hbox{$\left#1\vcenter
to1.111\ht\strutbox{}\right.\n@space$}}}

\newcommand{\prob}{{\rm Prob}}

\newcommand{\ie}{$\mbox{i.e.}$}

\newlength{\filength}
\newsavebox{\gcbox}
\sbox{\gcbox}{\framebox[\filength]{\rule{0ex}{2ex}}}

\newcommand{\qedblob}{\mbox{\rule[-1.5pt]{5pt}{10.5pt}}}
\def\literalqed{{\ \nolinebreak\hfill\mbox{\qedblob\quad}}}

\def\qed{\literalqed}

\newcommand{\singlespacing}{\let\CS=
\@currsize\renewcommand{\baselinestretch}{1}\tiny\CS}
\newcommand{\singlespacingplus}{\let\CS=
\@currsize\renewcommand{\baselinestretch}{1.25}\tiny\CS}
\newcommand{\doublespacing}{\let\CS=
\@currsize\renewcommand{\baselinestretch}{1.75}\tiny\CS}
\newcommand{\draftspacing}{\let\CS=
\@currsize\renewcommand{\baselinestretch}{2.0}\tiny\CS}

\def\zo{\{0,1\}}

\def\mapping{\rightarrow}

\def\calA{{\cal A}}

\newcommand{\zon}{\zo^n}

\makeatletter
\def\@listI{\leftmargin\leftmargini \parsep 4.5pt plus 1pt minus 1pt\topsep6pt plus 2pt minus 2pt \itemsep  2pt plus 2pt minus 1pt}

\let\@listi\@listI
\@listi
\makeatother

\author{ {Marius Zimand\/}
\thanks{  Department of Computer and Information Sciences, Towson University,
Baltimore, MD.; email: mzimand@towson.edu; http://triton.towson.edu/\~{ }mzimand.
The author is supported in part
by NSF grant CCF 0634830.}}

\author{
{Marius Zimand}
\thanks{ {  \tt  http://triton.towson.edu/\~{ }mzimand}.}}
\institute{
{Department of Computer and Information Sciences, Towson University,
Baltimore, MD, USA}
}
\date{ }

\title{Impossibility of independence amplification in Kolmogorov complexity theory}

\begin{document}

\pagestyle{plain}
\maketitle

\begin{abstract}
The paper studies randomness extraction from sources with bounded independence and the issue of independence amplification of sources, using the framework of Kolmogorov complexity. The dependency of strings $x$ and $y$ is ${\rm dep}(x,y) = \max\{C(x)  - C(x \mid y), C(y) - C( y\mid x)\}$, where $C(\cdot)$ denotes the Kolmogorov complexity. It is shown that there exists a computable Kolmogorov extractor $f$ such that, for any two $n$-bit strings with complexity $s(n)$ and dependency $\alpha(n)$, it  outputs a string of length $s(n)$ with complexity $s(n)- \alpha(n)$ conditioned by any one of the input strings. It is proven that the above are the optimal parameters a Kolmogorov extractor can achieve. It is shown that independence amplification cannot be effectively realized. Specifically, if (after excluding a trivial case) there exist computable functions $f_1$ and $f_2$  such that ${\rm dep}(f_1(x,y), f_2(x,y)) \leq \beta(n)$ for all $n$-bit strings $x$ and $y$  with ${\rm dep}(x,y) \leq \alpha(n)$, then $\beta(n) \geq \alpha(n) -  O(\log n)$.
\end{abstract}

\smallskip
\section{Introduction}
 Randomness extraction is an algorithmical process that improves the quality of a source  of randomness. A source of randomness can be modeled as a finite probability distribution, or a finite binary string, or an infinite binary sequence and the randomness quality is measured, respectively, by min-entropy, Kolmogorov complexity, and constructive Hausdorff dimension. All the three settings have been studied (the first one quite extensively).
 
 It is desirable to have an extractor that can handle very general classes of sources. Ideally, we would like to have an extractor that obtains random bits from a single defective source under the single assumption that there exists a certain amount of randomness in the source.  Unfortunately, this is not possible. In the case of finite distributions, impossibility results for extraction from a single source have been established by Santha and Vazirani~\cite{san-vaz:j:quasirand} and Chor and Goldreich~\cite{cho-gol:j:weaksource}. In the case of finite binary strings and Kolmogorov complexity randomness, Vereshchagin and Vyugin~\cite{ver-vyu:j:kolm} show that there exists strings $x$ with relatively high Kolmogorov complexity so that any string shorter than $x$ by a certain amount and which has small Kolmogorov complexity conditioned by $x$ (in particular any such shorter string effectively constructed from $x$) has small Kolmogorov complexity unconditionally. The issue of extraction from one infinite sequence has been first raised by Reimann and Terwijn~\cite{rei:t:thesis}, and after a series of partial results~\cite{rei:t:thesis,nie-rei:c:wtt-Kolm-increase,bie-dot-ste:j:haussdimension}, Miller~\cite{mill:j:KolmExtract} has given a strong negative answer, by constructing a sequence $x$ with ${\rm dim}(x) = 1/2$ such that, for any Turing reduction $f$, ${\rm dim}(f(x)) \leq 1/2$ (or $f(x)$ does not exist; ${\rm dim}(x)$ is the constructive Hausdorff dimension of the sequence $x$).
 
 Therefore, for extraction from a general class of sources, one has to consider the case of $t \geq 2$ sources, and in this situation, positive results are possible. Computable extractors from $t=2$ distributions with min-entropy $k = O(\log n)$ are constructed in~\cite{cho-gol:j:weaksource,dod-oli:c:extractor}. The construction of polynomial-time multisource extractors is a difficult problem. Currently, for $t=2$, the best results are by Bourgain~\cite{bou:j:multiextract} who achieves $k = (1/2 - \alpha)n$ for a small constant $\alpha$, and Raz~\cite{raz:c:multiextract} who achieves $k = {\rm polylog} n$ for one distribution and $k = (1/2 + \alpha)n$ for the other one. Polynomial-time extractors for $3$ or more distributions with lower values of $k$ for all distributions are constructed in~\cite{bar-imp-wig:c:multisourceext,bkssw:c:multisourceextract,raz:c:multiextract,rao-zuc:c:threesources,rao:c:multiextract}. Dodis et al.~\cite{dod-elb-oli-raz:c:twosourceextract} construct a polynomial-time $2$-source extractor for $k > n/2$, where the extracted bits are random conditioned by one of the sources. Kolmogorov extractors for $t \geq 2$ sources also exist. Fortnow et al.~\cite{fhpvw:c:extractKol} actually observe that any randomness extractor for distributions is a Kolmogorov extractor and Hitchcock et al.~\cite{hit-pav-vin:t:Kolmextraction} show  that a weaker converse
  holds, in the sense that any Kolmogorov extractor is a randomness condenser with very good parameters (``almost extractor''). For $t=2$, the works~\cite{zim:c:kolmlimindep,zim:c:genindepstringsCiE09} construct computable Kolmogorov extractors with better properties than those achievable by converting the randomness extractors from~\cite{cho-gol:j:weaksource} and ~\cite{dod-oli:c:extractor}. The case of infinite sequences is studied in \cite{zim:c:csr}, which shows that it is possible to effectively increase the constructive dimension if the input consists of two sources.
 
 All the positive results cited above require that the sources are independent. At a first glance, without independence, even the distinction between one and two (or more) sources is not clear. However, independence can be quantified and then we can consider two sources having bounded independence. It then becomes important to determine to what extent randomness extraction is possible from sources with a limited degree of independence and whether the degree of independence can be amplified. 
 
 We address these questions for the case of finite strings and Kolmogorov complexity-based randomness. The level of dependency of two strings is based on the notion of mutual information. The information that string $x$ has about string $y$ is $I(x:y) = C(y) - C(y \mid x)$, where $C(y)$ is the Kolmogorov complexity of $y$ and $C(y \mid x)$ is the Kolmogorov complexity of $y$ conditioned by $x$. By the symmetry of information theorem, $I(x:y) \approx I(y:x) \approx C(x)+C(y) - C(xy)$.\footnote{We use $\approx$, $\preceq$ and $\succeq$ for equalities and inequalities that hold within an additive error bounded by $O(\log n)$.} We define the dependency of strings $x$ and $y$ as ${\rm dep}(x,y) = \max\{I(x:y), I(y:x)\}$. Let $S_{k, \alpha}$ be the set of all pairs of strings $(x,y)$ such that $C(x) \geq k$, $C(y) \geq k$ and ${\rm dep}(x,y) \leq \alpha$. A Kolmogorov extractor for the class of sources $S_{k,\alpha}$ is a function $f: \zon \times \zon \mapping \zo^m$ such that for all $(x,y) \in S_{k, \alpha}$, $C(f(x,y))$ is ``close'' to $m$. In other words, if we define the randomness deficiency of a string $z$ as $|z| - C(z)$, we would like that the randomness deficiency of $f(x,y)$ is small. 
 Our first result shows that the randomness deficiency of $f(x,y)$ cannot be smaller than essentially the dependency of $x$ and $y$. 
 \smallskip
 
 \emph{Result 1} (informal statement; see full statement in Theorem~\ref{t:twodependentsources}). There exists no computable function $f$ with the property that, for all $(x,y) \in S_{k, \alpha}$, the randomness deficiency of $f(x,y)$ is less than $\alpha - \log n - O(\log \alpha)$. This holds true even for high values of $k$ such as $k \succeq n - \alpha$. The only condition is that $m \geq \alpha$ ($m$ is the length of the ouput of $f$).
 \smallskip
 
We observe that the similar result holds for the case of finite distributions. Let $S_{k,\alpha}$ be the set of all random variables over $\zon$ that have min-entropy at least $k$ and dependency at most $\alpha$. (The min-entropy of $X$ is $H_{\infty}(X) = \min_{a \in \zon, X(a) >0} \log (1/\prob[X=a])$ and the dependency of $X$ and $Y$ is $H_\infty(X) + H_\infty(Y) - H_{\infty}(X,Y)$.) Then, for every $\alpha$ and $m \geq \alpha$ and for every function $f: \zon \times \zon \mapping \zo^m$ (even non-computable), there exists $(X,Y) \in S_{k, \alpha}$ with dependency at most $\alpha$ and min-entropy of $f(X,Y)$ at most $m - \alpha$.

Our next result (and the main technical contribution of this paper) is a positive one. Keeping in mind \emph{Result 1}, the best one can hope for is a Kolmogorov extractor that from any strings $x$ and $y$ having dependency at most $\alpha$ obtains a string $z$ whose randomness deficiency is $\approx \alpha$. We show that this is possible in a strong sense.
\smallskip

\emph{Result 2} (informal statement; see full statement in Theorem~\ref{t:extractor}). For every $k > \alpha$, there exists a computable function $f:\zon \times \zon \mapping \zo^m$, where $m \approx k$, and such that for every $(x,y) \in S_{k, \alpha}$, $C(f(x,y) \mid x) = m - \alpha - O(1)$ and $C(f(x,y) \mid y) = m - \alpha - O(1)$.
\smallskip

Thus, optimal Kolmogorov extraction from sources with bounded independence can be achieved effectively and in a strong form. Namely, the randomness deficiency of the extracted string $z$ is minimal (\ie, within an additive constant of $\alpha$) even conditioned by any one of the input strings and furthermore the length of $z$ is maximal.
In~\cite{zim:c:kolmlimindep} a similar but weaker theorem has been established. The difference is that in~\cite{zim:c:kolmlimindep} the length of the output is only $\approx k/2$ and $k$ has to be at least $2 \alpha$.
The proof method of \emph{Result 2} extends the one used in~\cite{zim:c:kolmlimindep} in a non-trivial way (the  novel technical ideas are described in Section~\ref{s:proofoverview}).
We note that the Kolmogorov extractor that can be obtained from the randomness extractor from~\cite{dod-oli:c:extractor} using  the technique in~\cite{fhpvw:c:extractKol} would have weaker parameters (more precisely, the output length would be $m \approx k - 2 \alpha$).
 
The dependency of two strings $x$ and $y$ is another measure of the non-randomness in $(x,y)$ considered as a joint source. Similarly to Kolmogorov extractors that reduce randomness deficiency, it would be desirable to have an algorithm that reduces dependency (equivalently, amplifies independence). The main result of the paper shows that effective independence amplification is essentially impossible.
We say that two functions $f_1, f_2 : \zon \times \zon \mapping \zo^{l(n)}$ amplify independence from level $\alpha(n)$ to level $\beta(n)$ (for $\beta(n) < \alpha(n)$) if ${\rm dep}(f_1(x,y), f_2(x,y)) \leq \beta(n)$ whenever ${\rm dep}(x,y) \leq \alpha(n)$. Note that this is trivial to achieve if $f_1(x,y)$ or $f_2(x,y)$ have Kolmogorov complexity at most $\beta(n)$. Therefore, we also request that $f_1(x,y)$ and $f_2(x,y)$ have Kolmogorov complexity at least $\beta(n) + c \log n$, for some constant $c$. However, as a consequence of \emph{Result 1} and \emph{Result 2}, this is impossible for any reasonable choice of parameters.
\smallskip

\emph{Result 3} (informal statement; see full statement in Theorem~\ref{t:impossamplific}). Let $f_1$ and $f_2$ be computable functions such that for all $(x,y) \in S_{k,\alpha}$,
${\rm dep}(f_1(x,y), f_2(x,y)) \leq \beta(n)$ (and $C(f_1(x,y)) \succeq \beta(n), C(f_2(x,y)) \succeq \beta(n)$). Then $\beta(n) \succeq \alpha(n)$. This holds true for any $\alpha(n) \preceq n/2$ and any $k \preceq n - \alpha(n)$.
\smallskip
 
{\it Discussion of some technical aspects.} As it is typically the case in probabilistic analysis, handling sources with bounded independence is difficult. In this discussion, an $(n,k)$ source is a random variable over $\zon$ with min-entropy $k$. Chor and Goldreich~\cite{cho-gol:j:weaksource} show that a random function starting from any two independent sources of type $(n,k)$  extracts $\approx k/3$ bits that are close to random. Dodis and Oliveira~\cite{dod-oli:c:extractor} using a more refined probabilistic analysis (based on a martingale construction) show the existence of an extractor 
that from two independent sources $X$ and $Y$ of type $(n,k_1)$ and respectively $(n, k_2)$ obtains $\approx k_1$ bits that are close to random even conditioned by $Y$. Both constructions use in an essential way the independence of the two input distributions. The independence property allows one to reduce the analysis to the simpler case in which the two input distributions are so called \emph{flat distributions}. A flat distribution  with min-entropy $k$ assigns equal  probability mass  to a subset of size $2^k$ of $\zon$ and probability zero to the elements outside this set. Extractors that extract from flat distributions admit a nice combinatorial description. Namely, an extractor $E : \zon \times \zon \mapping \zo^m$ for two flat distributions $X, Y$ with min-entropy $k$ corresponds to an $N$-by-$N$ table (where $N = 2^n$) whose cells are colored with $M$ colors (where $M = 2^m$) that satisfy the following balancing property: For any set of colors $A \subseteq [M]$ and for any $K$-by-$K$ subrectangle of the table (where $K = 2^k$), the number of $A$-colored cells is close to $|A|/M$. Such tables can be obtained with the probabilistic method.

If the two input distributions are not independent, then the reduction to flat distributions is not known to be possible and the above approach fails. This is why almost all of the currently known randomness extractors (whether running in polynomial time, or merely computable) assume that the weak sources are perfectly independent (one exception is the paper~\cite{tre-vad:c:psamplextractor}).

In this light, it is surprising that Kolmogorov extractors for input strings that are not fully independent (actually with arbitrarily large level of dependency) can be obtained via balanced tables, as we do in this paper. This approach succeds because the Kolmogorov complexity-based analysis views the level of independence of sources as just another parameter and there is no need for any additional machinery to handle sources that are not fully independent. We believe (based on some partial results) that Kolmogorov complexity is a useful tool not only for analyzing  Kolmogorov extractors but also for circumventing some of the technical difficulties  in the investigation of multi-source extractors for sources with bounded independence.

 \section{Preliminaries}

We work over the binary alphabet $\{0,1\}$; $\nat$ is the set of natural numbers. A string $x$ is an element of $\{0,1\}^*$; $|x|$ denotes its length; $\zo^n$ denotes the set of strings of length $n$; $|A|$ denotes the cardinality of a finite set $A$; for $n \in \nat$, $[n]$ denotes the set $\{1,2, \ldots, n\}$. We recall the basics of (plain) Kolmogorov complexity (for an extensive coverage, the reader should consult one of the monographs by Calude~\cite{cal:b:infandrand}, Li and Vit\'{a}nyi~\cite{li-vit:b:kolmbook}, or Downey and Hirschfeldt~\cite{dow-hir:b:algrandom}; for a good and concise introduction, see Shen's lecture notes~\cite{she:t:kolmnotes}). Let $M$ be a standard Turing machine. For any string $x$, define the \emph{(plain) Kolmogorov complexity} of $x$ with respect to $M$, as 
$C_M(x) = \min \{ |p| \mid M(p) = x \}$.
 There is a universal Turing machine $U$ such that for every machine $M$ there is a constant $c$ such that for all $x$,
$C_U(x) \leq C_M(x) + c$.
We fix such a universal machine $U$ and dropping the subscript, we let $C(x)$ denote the Kolmogorov complexity of $x$ with respect to $U$. We also use the  concept of conditional Kolmogorov complexity. Here the underlying machine is a Turing machine that in addition to the read/work tape which in the initial state contains the input $p$, has a second tape containing initially a string $y$, which is called the conditioning information. Given such a machine $M$, we define the Kolmogorov complexity of $x$ conditioned by $y$ with respect to $M$ as 
$C_M(x \mid y) = \min \{ |p| \mid M(p, y) = x \}$.
There exist  universal machines of this type and they satisfy the relation similar to the above, but for conditional complexity. We fix such a universal machine $U$, and dropping the subscript $U$, we let $C(x \mid y)$ denote the Kolmogorov complexity of $x$ conditioned by $y$ with respect to $U$. 

There exists a constant $c_U$ such that for all strings $x$, $C(x) \leq |x| + c_U$. Strings $x_1, x_2, \ldots, x_k$ can be encoded in a self-delimiting way (\ie, an encoding from which each string can be retrieved) using $|x_1| + |x_2| + \ldots + |x_k| + 2 \log |x_2| + \ldots + 2 \log |x_k| + O(k)$ bits. For example, $x_1$ and $x_2$ can be encoded as $\overline{(bin (|x_2|)} 01 x_1 x_2$, where $bin(n)$ is the binary encoding of the natural number $n$ and, for a string $u = u_1 \ldots u_m$, $\overline{u}$ is the string $u_1 u_1 \ldots u_m u_m$ (\ie, the string $u$ with its bits doubled).

For every sufficiently large $n$ and $k \leq n$, for every $n$-bit string $y$,
$2^{k-2\log n} < |\{x \in \zo^n \mid C(x \mid y) \leq k\}| < 2^{k+1}$.

The Symmetry of Information Theorem~\cite{zvo-lev:j:kol} states that for any two strings $x$ and $y$, 
\begin{itemize}
\item[(a)] $C(xy) \leq C(y) + C(x \mid y) + 2 \log C(y) +O(1)$.
\item[(b)] $C(xy) \geq C(x) + C(y \mid x) - 2 \log C(xy) - 4 \log \log C(xy) - O(1)$.
\item[(c)] If $|x| = |y| = n$, $C(y) - C(y\mid x) \geq C(x) - C(x \mid y) - 5 \log n$
\end{itemize}

For integers $m \leq n$, let $b(n,m) = {n \choose 0} + {n \choose 1} +  \ldots + {n \choose m}$.
Note that $m (\log n - \log m) < \log b(n,m) < m (\log n - \log m) + m \log e + \log (1+m)$ (since $(n/m)^m < {n \choose m} < (en/m)^m$). 

All the Kolmogorov extractors will be ensembles of functions $f = (f_n)_{n \in \nat}$, of type $f_n: (\zon)^t \mapping \zo^{m(n)}$. The parameter $t$ is constant and indicates the number of sources (in this paper we only consider $t=1$ and $t=2$). For readability, we usually drop the subscript and the expression ``function $f : \zo^n \mapping \zo^m$ ...'' is a substitute for ``ensemble $f = (f_n)_{n \in \nat}$, where $f_n :\zon \mapping \zo^{m(n)}$,  ...''

We say that an ensemble of functions $f = (f_n)$ is computable with advice $k(n)$, if for every $n$ there exists a string $p$ of length at most $k(n)$ such that $U(p,1^n)$ outputs the table of the function $f_n$. 

We use the following standard version of the Chernoff bounds. Let $X_1, \ldots, X_n$ be independent random variables that take the values $0$ and $1$, let $X= \sum X_i$ and let $\mu$ be the expected value of $X$. Then, for any $ 0 < d \leq 1$, $\prob[X > (1+d)\mu] \leq e^{-d^2 \mu/3}$.

\subsection{Limited Independence}
\label{s:indep}

\begin{definition} 
\label{d:indep}
\begin{itemize}
\item[(a)] The dependency of two strings $x$ and $y$ is ${\rm dep}(x,y) = \max\{C(x) - C(x \mid y), C(y) - C(y \mid x)\}$.
\item[(b)] Let $d: \nat \mapping \nat$. We say that strings $x$ and $y$ have dependency at most $d(n)$ if ${\rm dep}(x,y) \leq d(\max(|x|, |y|))$.
\end{itemize}
\end{definition}
The Symmetry of Information Theorem implies that
\[
|{\rm dep}(x,y) - (C(x) - C(x \mid y))| \leq O(\log(C(x)) + \log (C(y))).
\]
If the strings $x$ and $y$ have length $n$, then
\[
|{\rm dep}(x,y) - (C(x) - C(x \mid y))| \leq 5 \log n.
\]

\section{Limits on Kolmogorov complexity extraction}
\subsection{Limits on extraction from one source}
We first show that for any single-source function computable with small advice there exists an input with high Kolmogorov complexity whose image has low Kolmogorov complexity.
\begin{proposition}
\label{p:onesourceimposs}
Let $f: \zo^n \mapping \zo^m$ be a function computable with advice $k(n)$. There exists $x \in \zo^n$ with $C(x) \geq n-m$ and $C(f(x)) \leq k(n) + \log n + 2 \log \log n + O(1)$.
\end{proposition}
\begin{proof}
Let $z$ be the most popular element in the image of $f$ (\ie, the element in $\zo^m$ with the largest number of preimages under $f$; if there is a tie, take $z$ to be the smallest lexicographically). Since $z$ can be described by the table of $f$ and $O(1)$ bits, it follows that $C(z) \leq k(n) + \log n  + 2 \log \log n + O(1)$. There are at least $2^{n-m}$ elements of $\zo^n$ mapping to $z$. Thus, there must be a string $x$ of complexity at least $n-m$ mapping to $z$.~\qed
\end{proof}

The following result is, in a sense, a strengthening of the previous proposition. It shows that there exists a string with relatively high Kolmogorov complexity, so that all functions computable with a given amount of advice fail to extract its randomness. We provide two incomparable combinations of parameters. Part~(b) is essentially a result of Vereshchagin and Vyugin~\cite{ver-vyu:j:kolm}.
\begin{theorem}
\label{t:nonextractablex}
For every $k$, every $n$, any computable function $m$:

(a) There exists a string $x \in \zo^n$ such that for every function $f: \zo^n \mapping \zo^{m}$ that is computable with advice $k = k(n)$,
\begin{itemize}
	\item[(1)] $C(x) > n- \log b(M,K) \geq n - K(m-k + O(1))$, where $M= 2^m, K = 2^{k+1}-1$, and
	\item[(2)] $C(f(x)) < 2k + 2 \log k + \log n + 2 \log \log n + O(1)$ or $f(x)$ is not defined.
\end{itemize}

and

(b) There exists a string $x \in \zo^n$ such that for every function $f: \zon \mapping \zo^m$ that is computable with advice $k$,
\begin{itemize}
	\item[(1)] $C(x) > n- K \log (M+1) \approx n - Km$, where $M= 2^m, K = 2^{k+1}-1$, and
	\item[(2)] $C(f(x)) < k + \log n + 2 \log \log n + O(1)$ or $f(x)$ is not defined.
\end{itemize}
\end{theorem}
\begin{proof}
Let $f_i$, $i \in \{1, \ldots, K\}$ be the function computed by $U(p_i, 1^n)$, where $p_i$ is the $i$-th string in $\{0,1\}^{\leq k}$. We fix $n$ and let $m = m(n)$.

For each $x \in \zo^n$, consider the computations $f_1(x), f_2(x), \ldots,f_K(x)$. Some of them may not halt, and some of them may produce strings of length different from $m$. Let ${\rm Range}(x)$ be the set of strings of length $m$ that result from these computations.

We first prove (a). ${\rm Range}(x)$ has one of $b (M, K)$ possible values. It follows that there exists one set that is equal to ${\rm Range}(x)$ for at least 
$2^{n}/b(M,K)$ many strings $x \in \zon$. We say that such a set is \emph{frequent}. Consider all frequent sets and let $s$ be the maximum size of a frequent set, taken over all frequent sets. If we know $s$, we can enumerate all frequent sets of size $s$. Let $\{z_1, \ldots, z_s\}$,  be the first such set that appears in the enumeration. Note that each entry $z_i$  can be described by $s$, $n$, $k$, and $i \leq s$. We can represent $i$ by a string having length exactly $k + 1$ bits and this string will therefore also describe $k$. It follows that each such $z_i$ satisfies
\[
\begin{array}{ll}
C(z_i) & \leq k + \log n + \log s  + 2 (\log \log n + \log \log s) + O(1) \\
& \leq 2k + \log n + 2 \log \log n + 2 \log k  + O(1),
\end{array}
\]
where we have used the fact that $i \leq K$ and $s \leq K$.
The set $\{z_1, \ldots, z_k \}$ is equal to at least $2^{n}/b(M,K)$  Ranges. So there exists $x$ with $C(x) \geq n - \log b(M, K)$ such that ${\rm Range}(x) = (z_1, \ldots, z_s)$. This $x$ satisfies the requierements in the statement.~\qed
\smallskip

We now prove (b) (following~\cite{ver-vyu:j:kolm}).
The goal,  as before, is to produce a set that is equal to ${\rm Range}(x)$, for many $x \in \zo^n$. We can do this, avoiding the information $s$ used in the previous proof, by the following greedy algorithm. By dovetailing the computations $f_i(x)$, for
all $x \in \zo^n$ and $i \in [K]$, we start enumerating strings produced by these computations, of which we retain only those having length $m$. Let $T = 2^m + 1$. We start the enumeration till we find a string $z_1$ that appears in at least $2^n/T$ ranges. There may be no such $z_1$ and  we handle this situation later. We mark with ($1$) all Ranges that have been identified to contain $z_1$. In the second iteration, we restart the enumeration till we find a string $z_2 \not = z_1$ that belongs to at least $1/T$ fraction of Ranges marked with ($1$). We re-mark these Ranges with ($2$). In general, at iteration $i$, we find a string
$z_i$, different from $z_1, \ldots, z_{i-1}$, that belongs to at least a fraction $1/T$ of Ranges marked $(i-1)$. If we find such a $z_i$, we mark the Ranges that have been discovered to contain it with $(i)$.

We keep on doing this process till either (a) we have completed $K$ iterations and have obtained $K$ distinct strings $z_1, \ldots, z_K$ in $\zo^m$, or (b) at iteration $i$, the enumeration failed to produce $z_i$.

In case (a), the set $\{z_1, \ldots, z_K\}$ is equal to at least $2^n/T^K$ Ranges.

In case (b), the set $\{z_1, \ldots, z_{i-1}\}$ is a subset of at least $2^n/T^{i-1}$ Ranges, and for each other string $z \in \zo^m$, the set 
$\{z_1, \ldots, z_{i-1}, z\}$ is a subset of less than $2^n/T^i$ Ranges.
It follows that there exist at least $2^n/T^{i-1} - 2^m \cdot 2^n/T^i = 2^n/T^i$ Ranges that are equal to the set $\{z_1, \ldots, z_{i-1}\}$. 

To conclude, there exists a set $\{z_1, \ldots, z_s\}$, with $s \leq K$, that is equal to ${\rm Range}(x)$ for at least $2^n/ (2^m+1)^K$ strings $x \in \zo^n$. Therefore there exists such a string $x$ with $C(x) \geq n - K \log (2^m+1)$. Each element $z_i$ is described by $i \leq K$, $n$ and $k$. We represent $i$ on exactly $k+1$ bits and this also describes $k$. Therefore $C(z_i) \leq k + \log n + 2 \log \log n + O(1)$. The conclusion follows.~\qed
\end{proof}
\subsection{Limits on extraction from two sources}
The following theorem shows that there is no uniform function that from two sources $x$ and $y$ that are $\alpha$-dependent (\ie, ${\rm dep}(x,y) \succeq \alpha$), produces an output whose randomness deficiency is less than $\alpha - \log n - O(\log \alpha)$.
\begin{theorem}
\label{t:twodependentsources}
Let $f: \zo^n \times \zo^n \mapping \zo^m$ be a computable function and let $\alpha \in \nat$, $\alpha \leq m$. Then there exists a pair of strings $x \in \zo^n, y \in \zo^n$ such that
\[
\begin{array}{ll}
C(x \mid y) & \geq n - \alpha - 2\log n \\
C( y \mid x ) & \geq n - \alpha -  2\log n \\
C(f (x,y) ) & \leq m - \alpha + \log n + 2\log \alpha + O(1).
\end{array}
\] 
\end{theorem}
\begin{proof}
We consider first the case $m=\alpha$. Let $a$ be the most popular string in the image of $f$. Then $C(a) < \log n + O(1)$. Since $|f^{-1}(a)|$ has at least $2^{2n-m}$ elements, there exists
strings $x$ and $y$ in $\zo^n$ such that $(x,y) \in f^{-1}(a)$ and $C(xy) \geq 2n-m$. Since $C(xy) \leq C(x \mid y) + C(y \mid x) + 2\log n$ and $C(x) \leq n + O(1)$ and $C(y) \leq n + O(1)$, it follows that
$C(x \mid x) \geq n- m - 2\log n$ and $C(y \mid x) \geq n- m - 2\log n$. Also $C(f(x,y)) = C(a) < \log n + O(1)$.

If $m > \alpha$, take $g(x,y)$ the prefix of length $\alpha$ of $f(x,y)$. Then $C(f(x,y)) \leq C(g(xy)) + (m-\alpha) + 2 \log \alpha + O(1) < \log n + (m-\alpha) + 2\log \alpha + O(1)$, and the conclusion follows.~\qed
\end{proof}

The following is the analog of Theorem~\ref{t:twodependentsources} for distributions. 
\begin{theorem}
\label{t:impossdistributions}
Let $f: \zo^n \times \zo^n \mapping \zo^m$ be a function and let $\alpha \in \nat$, $\alpha \leq m$. Then there exists two random variables $X$ and $Y$ taking values in $\zo^n$ such that
\[
\begin{array}{ll}
H_\infty(X) & \geq n- \alpha \\
H_\infty(Y) & \geq n- \alpha \\
H_\infty(X,Y) & \geq 2n - \alpha \\
H_\infty(f(X,Y)) &  \leq  m - \alpha
\end{array}
\] 
\end{theorem}

\begin{proof}
Suppose first that $m = \alpha$. Let $a$ be the most popular string in the image of $f$. Then $|f^{-1}(a)| \geq 2^{2n-m}$. Take (arbitrarily) $B \subseteq f^{-1}(a)$ with $|B| = 2^{2n-m}$. Consider
$\mbox{LEFT-B}$ the multiset of $n$-bit prefixes of strings in $B$ and $\mbox{RIGHT-B}$ the multiset of $n$-bit suffixes of strings in $B$. The multiplicity of a string $x$ in $\mbox{LEFT-B}$ is equal to the number of strings in $B$ that have $x$ as their left half. Thus each string in $\mbox{LEFT-B}$ has multiplicity at most $2^n$. Counting multiplicities $\mbox{LEFT-B}$ has $2^{2n-m}$ elements. Therefore $\mbox{LEFT-B}$ has at least $2^{n-m}$ distinct strings. The same holds for $\mbox{RIGHT-B}$. We take $X$ to be the random variable obtained by choosing uniformly at random one element in the multiset $\mbox{LEFT-B}$ and $Y$ is the random variable obtained by choosing uniformly at random one element in the multiset $\mbox{RIGHT-B}$. By the above discussion for each $x \in \zo^n$ and $y \in \zo^n$,
\[
\begin{array}{ll}
\prob[X=x] & \leq \frac{2^n}{2^{2n-m}} = \frac{1}{2^{n-m}}, \\
\prob[Y=y] & \leq \frac{2^n}{2^{2n-m}} = \frac{1}{2^{n-m}}, \\
\prob[X=x, Y=y] & \leq \frac{1}{2^{2n-m}}.
\end{array}
\]
Thus, $X$ and $Y$ satisfy the requirements, and $\prob[f(X,Y) = a] = 1$.

Suppose now that $m > \alpha$. We define $g, h :\zo^n \times \zo^n \mapping \zo^{\alpha}$ by
$g(x,y) = $ prefix of length $\alpha$ of $f(x,y)$ and $h(x,y) = $ suffix of length $m-\alpha$ of $f(x,y)$. Let $a \in \zo^{\alpha}$ and  the random variables $X$ and $Y$ defined as in the first part of the proof (\ie, the case $m=\alpha$) but with $g$ replacing $f$. Note that $\prob[g(X,Y) = a] = 1$. Let $b$ be a string in $\zo^{m-\alpha}$ such that $h^{-1}(b)$ has at least $2^{2n}/2^{m-\alpha}$ elements. Then
\[
\begin{array}{ll}
\prob [ f(X,Y) = ab] & = \prob[g(X,Y) = a, h(X,Y) = b] \\
&= \prob[h(X,Y) = b] \\
& \geq \frac{2^{2n}/2^{m-\alpha}}{2^{2n}} = 2^{-(m-\alpha)}.

\end{array}
\]
This concludes the proof.~\qed
\end{proof}

\if01
The following is a strenghtening of Theorem~\ref{t:twodependentsources}, in the spirit of Theorem~\ref{t:nonextractablex}. Its proof is omitted in this extended abstract.

\begin{theorem}
For every $n$, $k$, $\alpha$, $m \geq \alpha$, there exists a pair of strings $x \in \zo^n, y \in \zo^n$ such that
\[
\begin{array}{ll}
C(x) & \geq n - \alpha - k - O(\log n)\\
C( y \mid x ) & \geq n - \alpha -  k - O(\log n) \\
\end{array}
\] 
and for every function $f:\zon \times \zon \mapping \zo^m$ computable with advice $k$,
\[
C(f(x,y)) \leq m - \alpha + (3k + 2\log n + 2\log k + O(1)).
\]
\end{theorem}
\fi

\section{Kolmogorov complexity extraction}
We construct a Kolmogorov extractor that on input two $n$-bit strings with Kolmogorov complexity at least $s(n)$ and dependency at most $\alpha(n)$ outputs a string of length $\approx s(n)$ having complexity $\approx s(n) - \alpha(n)$ conditioned by any one of the input strings.
\subsection{Proof overview}
\label{s:proofoverview}
For an easier orientation in the proof, we describe the main ideas of the method. We also explain the non-trivial way in which the new construction extends the technique from the earlier works~\cite{zim:c:kolmlimindep} and~\cite{zim:c:genindepstringsCiE09}.  For readability, some details are omitted and some estimations are slightly imprecise.  Let us fix, for the entire discussion,  $x$ and $y$,  two $n$-bit strings with $C(x) \geq s(n)$ and $C(y) \geq s(n)$ and having dependency at most $\alpha(n)$.  We denote $N=2^n, M=2^m$ and $S=2^{s(n)}$. Let $B_x =\{ u \mid C(u) \leq C(x) \}$ and $B_y = \{ v \mid C(v) \leq C(y)\}$. An  $N$-by-$N$ table colored with $M$ colors is a function $T:[N] \times [N] \mapping [M]$. If we randomly color such a table $T$, with parameter $m \preceq 2 s(n)$, then, with high probability, no color appears in the $B_x \times B_y$ rectangle more than $2 \cdot (1/M)$ fraction of times (we say that a table that satisfies the above balancing property is balanced in $B_x \times B_y$). Clearly $(x,y) \in B_x \times B_y$ and in a table $T$ balanced in $B_x \times B_y$ there are at most $2 \cdot (1/M) \cdot |B_x| \times |B_y| \approx 2 \cdot (1/M) 2^{C(x)} 2^{C(y)} = 2^{C(x) + C(y) - m +1}$ entries with the color $z = T(x,y)$. Therefore $(x,y)$ is described by the color $z = T(x,y)$, the rank $r$ of the $(x,y)$ cell in the list of all $z$-colored cells in $B_x \times B_y$, by the table $T$, and by $O(\log n)$ additional bits necessary to enumerate the list. Thus, $C(xy) \leq C(z) + \log r + C(\mbox{table } T)+ O(\log n)$. By the above estimation, $\log r \approx C(x) + C(y) - m$. Also $C(xy) \geq C(x) + C(y) - {\rm dep}(x,y)$. Suppose that we are able to get a balanced table $T$ with $C(\mbox{table } T) = O(\log n)$, \ie, a table that can be described with $O(\log n)$ bits. Then we would get that $C(T(x,y)) = C(z) \geq m - {\rm dep}(x,y)$, which is our goal. How can we obtain $C(\mbox{table } T) = O(\log n)$? The normal approach would be to enumerate all possible $N$-by-$N$ tables with all possible colorings with $M$ colors and pick the first one that satisfies the balancing property. However, since $B_x$ and $B_y$ are only computably enumerable, we can never be sure that a given table has the balancing property. Therefore, instead of restricting to only $B_x$ and $B_y$, we require that a table $T$ should satisfy the balancing property for \emph{all} rectangles $B_1 \times B_2$ with sizes $|B_1| \geq S$ and  $|B_2| \geq S$, where $S = 2^{s(n)}$.  The simple probabilistic analysis involves only an additional union bound and carries over showing that such balanced tables exist at the cost that this time we need $m \preceq s(n)$.  Now we can pick  in an effective way the smallest (in some canonical order) table $T$ having the balancing property, because we can check the balancing property in an exhaustive manner (look at all $S \times S$-sized rectangles, etc.). Therefore this table $T$ can be described with $\log n + O(1)$ bits, as desired. In this way, from any $x$ and $y$, each having Kolmogorov complexity at least $s(n)$, we obtain $m \approx s(n)$ bits having Kolmogorov complexity $m - {\rm dep}(x,y)$. We reobtain $m \approx 2 s(n)$ if we change the balancing property and require that for any subset of colors $A \subseteq [M]$ of size $M/D$, for $D \approx 2^{\alpha(n)}$, for any rectangle $B_1 \times B_2$ with sizes $|B_1| \geq S$ and $|B_2| \geq S$, the fraction of $A$-colored cells in $B_1 \times B_2$  should be at most $2 \cdot (|A|/M) = 2 \cdot (1/D)$. Such a table can be obtained with $m \approx 2 s(n)$, and thus we can extract $\approx 2 s(n)$ bits having Kolmogorov complexity $\approx 2 s(n) - {\rm dep}(x,y)$, which is optimal.

Let us consider next the problem of extracting bits that are random even conditioned by $x$, and also conditioned by $y$. Suppose we use tables that satisfy the first balancing property. We focus on $B_x = \{u \mid C(u) \leq C(x) \}$ and we call a column $v$ \emph{bad} for a color $a \in [M]$ if the fraction of $a$-colored cells in the strip $B_x \times \{v\}$ of the table $T$ is more than $2 \cdot (1/M)$. The number of bad columns is less than $S$; otherwise the table would have an $S \times S$-sized rectangle that does not have the balancing property. Note that a bad column for a color $a$ can be described by the color $a$ and its rank in an enumeration of the columns that are bad for $a$ plus additional $O(\log n)$ bits. So if $v$ is a bad column, then $C(v) \preceq m + s(n) \approx 2 s(n)$. Therefore if $C(y) \succeq 2 s(n)$, $y$ is good for any color. An adaptation of the above argument shows that for $z = T(x,y)$ it holds that $C(x \mid y) \preceq C(z \mid y) + C(x) + C(y) - m$, which combined with $C(x \mid y) \succeq C(x) + C(y) - {\rm dep}(x,y)$, implies $C(z \mid y) \succeq m - {\rm dep}(x,y)$. The above holds only for $y$ with $C(y) \geq 2s(n)$ and since the probabilistic analysis  requires $m$ to be less than $s(n)$, it follows that the number of extracted bits (which is $m$) is less than half the Kolmogorov complexity of $y$.

The above technique was used in~\cite{zim:c:kolmlimindep} and in~\cite{zim:c:genindepstringsCiE09}. To increase the number of extracted bits, we introduce a new balancing property, which we dub \emph{rainbow balancing}. Fix some parameter $D$, which eventually will be taken such that $\log D \approx {\rm dep}(x,y)$. Let ${\cal A}_D$ be the collection of sets of colors $A \subseteq [M]$ with  size $|A| \approx M/D$. Let $B_1 \subseteq [N]$ be a set of size a multiple of $S$, let $\overline{v} = \{v_1 < v_2 \ldots < v_S\}$ be a set of $S$ columns, and let $\overline{A} = (A_1, \ldots, A_S)$ be a tuple with each $A_i$ in ${\cal A}_D$. We say that a cell $(u,v_i)$ such that $T(u,v_i) \in A_i$ is properly colored with respect to $\overline{v}$ and $\overline{A}$.  Finally we say that a table $T: [N] \times [N] \mapping [M]$ is $(S,D)$-rainbow balanced if for every $B_1$, every $\overline{v}$, and every $\overline{A}$, the fraction of cells in $B_1 \times \overline{v}$ that are properly colored with respect to $\overline{v}$ and $\overline{A}$ is at most $2 \cdot (1/D)$. The probabilistic method shows that such tables exist provided $m \preceq s(n)$ and $\log D \preceq s(n)$. Since the rainbow balancing property can be effectively checked, there is an $(S,D)$-rainbow balanced table $T : [N] \times [N] \mapping [M]$ that can be described with $\log n + O(1)$ bits and $m \approx s(n)$ and $\log D \approx s(n)$. Let $z = T(x,y)$ and suppose that $C(z \mid y) < m - t$, where $t = \alpha(n) - c \log n$, for some constant $c$ that will be defined later (in the actual proof we do a tighter analysis and we manage to take $t = \alpha(n) - O(1)$). For each $v$, let $A_v = \{w \in [M] \mid C(w \mid v) < m - t\}$. For $ \log D \approx \alpha(n) + c \log m$, it holds that $A_v \in {\cal A}_D$
 for all $v$. Let us call a column $v$ \emph{bad} if the fraction of cells in $B_x \times \{v\}$ that are $A_v$-colored is larger than $2 \cdot (1/2^t)$. Analogously to our earlier discussion, the number of bad columns is less than $S$ and from here we infer that if $v$ is a bad column, then $C(v) \preceq s(n)$. Since $C(y) \geq s(n)$, it follows that $y$ is a good column. Therefore the fraction of cells in the $B_x \times \{y\}$ strip of the table $T$ that have a color in $A_y$ is at most $2\cdot (1/2^t)$. Since $(x,y)$ is one of these cells, it follows that, given $y$, $x$ can be described by the rank $r$ of $(x,y)$ in an enumeration of the $A_y$-colored cells in the strip $B_x \times \{y\}$, a description of the table $T$, and by $O(\log n)$ additional bits necessary for doing the enumeration. Note that there are at most $2 \cdot (1/2^t) \cdot |B_x| \approx 2^{-t+1} \cdot 2^{C(x)}$ cells in $B_x \times \{y\}$ that are $A_y$-colored and, therefore, $\log r \leq C(x) - t + 1$. From here we obtain that $C(x \mid y) \leq C(x) - t + 1 + O(\log n) = C(x) - \alpha(n) - c \log n +O(\log n)$. Since $C(x \mid y) \geq C(x) - \alpha(n)$, we obtain a contradiction for an appropriate choice of the constant $c$. Consequently $C(z \mid y) \geq m - t = m - \alpha(n) + c \log n$. Similarly, $C(z \mid x) \geq m - \alpha(n) + c \log n$. Thus we have extracted $m \approx s(n)$ bits that have Kolmogorov complexity $\approx m - \alpha(n)$ conditioned by $x$ and also conditioned by $y$.

\subsection{Construction of the Kolmogorov extractor}
\label{s:construction}
For $n$ and $m$ natural numbers, let $N = 2^n$ and $M = 2^m$. Henceforth, we identify $\zo^n$ with $[N]$ and $\zo^m$ with $[M]$.

We consider functions of the form $T: [N] \times [N] \mapping [M]$, which we view as $N$-by-$N$ tables whose cells are colored with colors in $[M]$. Let $S$ and $D$ be parameters with $S \leq N$ and $D \leq M$.

Let ${\calA}_D = \{A \mid A \subseteq [M], (M/D) \leq |A| \leq (M/D)m^2 \}$. Thus, the elements of $\calA_D$ are those sets of colors having at least $M/D$ colors and not much more than that. 

Let $B_2 \subseteq [N]$ be a subset of size $S$; we name its elements $B_2 = \{v_1 < v_2 < \ldots < v_S\}$. We view $B_2$ as a set of columns in the table. Let $(A_1, \ldots, A_S) \in (\calA_D)^S$.  The cell $(u,v_i) \in [N] \times B_2$
is \emph{properly colored} with respect to the columns in $B_2$ and $(A_1, \ldots, A_S)$ if $T(u,v_i) \in A_i$. A similar notion of a cell being properly colored with respect to \emph{rows} in a set $B_1 \subseteq [N]$ will also be used.
\begin{definition}
\label{d:tables}
A table $T: [N] \times [N] \mapping [M]$ is $(S,D)$-rainbow balanced if
\begin{itemize}
\item[(a)]
\begin{itemize}
	\item for all $B_1 \subseteq [N]$ of size $k \cdot S$ for some positive natural number $k$,
	\item for all $B_2 \subseteq [N]$ of size $S$,
	\item for all $(A_1, \ldots, A_S) \in (\calA_D)^S$,
	
\end{itemize}
it holds that the number of cells in $B_1 \times B_2$ that are properly colored with respect to columns $B_2$ and $(A_1, \ldots, A_S)$ is at most
\[
2m^2 \frac{|B_1| \cdot |B_2|}{D},
\]
and 
\item[(b)]
if the similar relation holds if we switch the roles of $B_1$ and $B_2$.

\end{itemize}
\end{definition}
\begin{lemma}
\label{l:tables}
If $S \geq 12D + 3 (1 + \ln D) M m^2 + 6D \ln (N/S)$, there exists a table $T:[N] \times [N] \mapping [M]$ that is $(S,D)$-rainbow balanced.
\end{lemma}
\begin{proof}
We use the probabilistic method. We show that a randomly colored table fails with probability  $< 1/2$ to satisfy the proper coloring property with respects to columns (property (a) in definition~\ref{d:tables}). A similar calculation shows the similar fact about proper coloring with respects to rows (property (b) in definition~\ref{d:tables}). Therefore we can conclude that a $(S,D)$-rainbow balanced table exists.

Observe that it is enough to consider sets $B_1$ of size exactly $S$ (because a set of size $kS$ can be broken into $k$ sets of size $S$ and if each smaller set satisfies the property, then the larger set will satisfy it as well).

Therefore, let us fix $B_1$ and $B_2$ subsets of $[N]$ of size $S$, let $B_1=\{u_1 < \ldots < u_S\}$ and  $B_1=\{v_1 < \ldots < v_S\}$. We fix $(A_1, \ldots, A_S) \in (\calA_D)^S$.

Let $X_{i,j}$ be the random variable which is $1$ if the cell$(u_i, v_j)$ is properly colored with respect to columns in $B_2$  and $(A_1, \ldots A_S)$ (\ie, $T(u_i, v_j) \in A_j$), and $0$ otherwise.
Then 
\[
\prob[X_{i,j} = 1] = \frac{|A_j|}{M} = \mu_j \in [1/D, m^2/D].
\]
Let $X = \sum_{i \in B_1, j \in B_2} X_{i,j}$. Then
\[
\mu = E[X] = \sum_{j \in B_2} \sum_{i \in B_1} E[X_{i,j}] = \sum_{j \in B_2}S \cdot \mu_j \in [S^2/D, S^2 \cdot m^2/D].
\]
By the Chernoff bounds,
\[
\prob[X \geq 2 \mu] \leq e^{-(1/3) \mu} \leq e^{-(1/3) (S^2/D)}.
\]
It follows that
\[
\prob[X \geq 2 \frac{S^2 m^2}{D}] \leq \prob[ X \geq 2 \mu] \leq e^{-(1/3) (S^2/D)}.
\]
We next take the union bound over all possible choices of $(A_1, \ldots, A_S) \in (\calA_D)^S$, and all possible choices of $B_1$ and $B_2$ subsets of $[N]$ of size $S$.

For $T \in [M/D, M\cdot m^2/D]$, the number of sets in $[M]$ of size $T$ is ${M \choose T} \leq (\frac{eM}{T})^T = e^T \cdot e^{T \ln (M/T)} \leq e^{T + T \ln D}$.
So the number of subsets of $[M]$ with sizes between $M/D$ and $M \cdot m^2/D$ is at most
\[
\sum_{T=M/D}^{M\cdot m^2/D} e^{T(1+ \ln D)}.
\]
Denoting $q = e^{(1+ \ln D)}$, the above sum is
\[
\begin{array}{ll}
\sum_{T=M/D}^{M\cdot m^2/D} e^{T(1+ \ln D)} & = q^{(M/D)} + q^{(M/D) + 1} + \ldots + q^{(M/D)m^2} \\
& = q^{(M/D)} \frac{q^{(M/D)(m^2 - 1) + 1} - 1}{q-1} \\
& < q^{(M/D)} \cdot q^{(M/D)m^2} \cdot q^{-(M/D)} \cdot \frac{q}{q-1} \\
& < 2 q^{(M/D)\cdot m^2} = 2 \cdot e^{(1 + \ln D) \cdot (M/D)\cdot m^2}.

\end{array}
\]
So the number of tuples $(A_1, \ldots, A_S) \in (\calA_D)^S$ is less than
$2^S \cdot e^{S \cdot (1 + \ln D) \cdot (M/D)\cdot m^2}$.

The number of ways of choosing $B_1$ and $B_2$ is
\[
{N \choose S} \cdot {N \choose S} \leq \big( \frac{eN}{S}\big)^{2S} = e^{2S + 2S \ln (N/S)}.
\]
For the union bound to give a probability $\leq e^{-1} < 1/2$ we need
\[
(1/3)(1/D)S^2 \geq S + S(1 + \ln D)(M/D)m^2 + 2S + 2S \ln(N/S) + 1,
\]
which holds true if the parameters satisfy the hypothesis.~\qed
\end{proof}

\begin{theorem}
\label{t:extractor}
For any computable functions $s(n)$ and $\alpha(n)$ with $n \geq s(n) \geq \alpha(n) + 7 \log n + O(1)$, for every computable function $m(n)$ with $m(n) \leq s(n)-7 \log n$, there exists a computable function $E: \zon \times \zon \mapping \zo^{m(n)}$, such that for all $x$ and $y$ in $\zo^n$ if
\begin{itemize}
	\item[(i)] $C(x) \geq s(n), C(y) \geq s(n)$,
	\item[(ii)] $C(x) - C(x \mid y) \leq \alpha(n)$ and $C(y) - C(y \mid x) \leq \alpha(n)$,

\end{itemize}
then
\begin{itemize}
	\item [(1)] $C(E(x,y) \mid x) \geq m - \alpha(n) - O(1)$,
	\item [(2)] $C(E(x,y) \mid y) \geq m - \alpha(n) - O(1)$.
\end{itemize}

\end{theorem}
\begin{proof}
The construction depends on a constant $C$ that will be determined later. Let $s = \lfloor s(n) - 3 \log n \rfloor$, $S = 2^s$, $D = 2^{\alpha(n) + C + 2 \log m}$ and $t = \alpha(n) + C$.

By Lemma~\ref{l:tables} there exists $T: [N] \times [N] \mapping [M]$ an $(S,D)$-rainbow balanced table.
We consider the smallest (in some canonical order) such table $T$ and define $E(x,y)$ to be $T(x,y)$. Thus, the table $T$ can be described with $\log n + O(1)$ bits.

Let us fix $x$ and $y$ with $C(x) = t_1 \geq s(n)$, $C(y) = t_2 \geq s(n)$ and ${\rm dep}(x,y) \leq \alpha(n)$.

Let $z = T(x,y)$. We prove that $C(z \mid y) \geq m - \alpha(n) - C = m-t$ and $C(z \mid x) \geq m - \alpha(n) - C = m-t$. Actually we show just the first relation (the second one is similar).

Suppose $C(z \mid y) < m-t$.

Let $B_1 = \{u \in \zon \mid C(u) \leq t_1\}$ and $B_2 = \{v \in \zon \mid C(v) \leq t_2\}$.
We have $|B_1| < 2^{t_1+1}$ and $|B_2| \leq 2^{t_2+1}$. Take supersets $B_1' \supseteq B_1$ and
$B_2' \supseteq B_2$ with $|B_1'| = 2^{t_1+1}$ and $|B_2'| = 2^{t_2+1}$ (and $B_1'$, $B_2' \subseteq [N]$). Note that the sizes of
$B_1'$ and $B_2'$ are exact multiples of $S$.

For each $v \in \zo^n$, let $A_v = \{w \in \zo^m \mid C(w \mid v) < m-t\}$. Note that
$2^{m-t - 2\log m} \leq |A_v| < 2^{m-t}$ and thus
$M/D \leq |A_v| \leq M\cdot m^2/D$.
In other words, for all $v \in \zon$, $A_v \in \calA_D$.

We say that $v \in \zon$ is a \emph{bad column} if the number of cells in $B_1 \times \{v\}$ that are $A_v$-colored is at least $2 \cdot \frac{|B_1'|}{2^t}$.

Since $B_1 \subseteq B_1'$, if $v$ is a bad column, the number of $A_v$-colored cells in $B_1' \times \{v\}$ is also at least $2 \cdot \frac{|B_1'|}{2^t}$. It follows that the number of bad columns is less than $S$. Otherwise, there would be $S$ columns $v_1, \ldots, v_S$ that fail to satisfy (a) in Definition~\ref{d:tables} for $B_1'$ and the tuplet of colors $(A_{v_1}, \ldots, A_{v_S})$, and this is not possible because the table $T$ is rainbow balanced. 

The set of bad columns can be enumerated if we are given $t_1$, $m-t$ and the table $T$. Therefore, if $v$ is a bad column, then $v$ can be described by its rank in the enumeration of the bad columns and by the information needed for the enumeration. Note that from $n$, we can calculate the table $T$ and $m-t$.  Therefore,
\[
\begin{array}{ll}
C(v) & \leq \log(S) + \log(t_1)  + \log n + 2\log \log t_1 +  2 \log \log n + O(1) \\
& < s + 3 \log n.
\end{array}
\]
Since $C(y) \geq s(n) = s + 3 \log n$, $y$ is a good column.

Let $G$ be the positions in the strip $B_1 \times \{y\}$ that are $A_y$-colored. Formally,
$G = {\rm proj}_1(T^{-1}(A_y) \cap (B_1 \times \{y\})$.
By assumption, $x$ belongs to the set $G$.
Since $y$ is a good column, 
\[
|G| \leq 2 \frac{|B_1'|}{2^t} = \frac{2^{t_1+2}}{2^t}.
\]
The set $G$ can be enumerated given $y$, $t_1$, $m-t$ and the table $T$. Thus, given $y$, $x$ can be described by its rank in the enumeration of $G$ and by the information needed for the enumeration. This information is given as follows. We give the constant $C$ and the rank of $x$ written on exactly $t_1+2-t$. Note that from $y$, whose length is $n$, we can calculate the table $T$ and $m$ and $t$.  Thus, from the given information, we can reconstruct $t_1$. 
Therefore,
\[
\begin{array}{ll}
C(x\mid y) & \leq t_1+2 -t + \log C + 2 \log \log C + O(1) \\
& < t_1 - t + \log C + 2 \log \log C + O(1),
\end{array}
\]
where the constant in $O(1)$ does not depend on $C$.
On the other hand, since $x$ and $y$ are at most $\alpha(n)$-dependent,
\[
C(x\mid y) \geq t_1 - \alpha(n).
\]
Combining the last two inequalities, it follows that $t < \alpha(n) + \log C + 2 \log \log C + O(1)$, which contradicts that $t=\alpha(n) + C$. (for an appropriate choice of $C$)~\qed
\end{proof}

\section{Impossibility of independence amplification}

The dependence of strings $x$ and $y$ is given by ${\rm dep}(x,y) = C(x) + C(y)  - C(xy)$. The smaller ${\rm dep}(x,y)$ is, the more independent the strings $x$ and $y$ are. Thus, amplifying independence amounts to reducing dependence. An effective dependence reducer would consist of two computable functions $f_1$ and $f_2$ that for two functions $\alpha(n) > \beta (n)$ guarantee that for all $x,y$ of length $n$,
\begin{equation}
\label{e:reducer}
\mbox{dep}(x,y) \leq \alpha(n) \Rightarrow \mbox{dep}(f_1(x,y), f_2(x,y)) \leq \beta(n).
\end{equation}

Note that, since ${\rm dep}(u,v) \leq \beta(n)$ whenever $C(u) \leq \beta(n)$ or $C(v) \leq \beta(n)$, dependency reduction would be achieved by two functions that simply output strings with Kolmogorov complexity $\leq \beta(n)$. To avoid this trivial and non-interesting type of dependency reduction, we require that, in addition to requierement~(\ref{e:reducer}), $C(f_1(x,y)) \succeq \beta(n)$ and $C(f_2(x,y)) \succeq \beta(n)$. More precisely, we seek two computable functions 
$f_1 : \zon \times \zon \mapping \zo^{l(n)}$ and $f_2 : \zon \times \zon \mapping \zo^{l(n)}$ that satisfy the following DEPENDENCY REDUCTION TASK.
\medskip

\fbox{
\parbox{11cm}{
\vbox{
DEPENDENCY REDUCTION TASK for parameters $ \alpha(n), \beta(n)$,  $s(n)$, $l(n)$, and $a$.
\smallskip

For all $x \in \zon$, $y \in \zon$ with ${\rm dep}(x,y) \leq \alpha(n)$, $C(x) \geq s(n)$ and $C(y) \geq s(n)$ the following should hold:
\begin{enumerate}
	\item ${\rm dep}(f_1(x,y), f_2(x,y)) \leq \beta(n)$,
	\item $C(f_1(x,y)) \geq \beta(n) + a \cdot \log n$ and $C(f_2(x,y)) \geq \beta(n) + a \cdot \log n$.
\end{enumerate}
}}}
\smallskip

We show that effective independence amplification is essentially impossible.
\begin{theorem}
\label{t:impossamplific}
Let $\alpha(n)$ be a function such that $\alpha(n) \leq n/2 - 5\log n$ and let $\beta(n) = \alpha(n) - \log n - 3 \log \alpha(n)$. Let $s(n)$ be a function such that $s(n) \leq n - \alpha(n) - 2 \log n - O(1)$ and let $l(n)$ be a function such that $l(n) \geq \beta(n) + 8 \log n$.

There are no computable functions $f_1: \zon \times \zon \mapping \zo^{l(n)}$ and $f_2: \zon \times \zon \mapping \zo^{l(n)}$ satisfying the DEPENDENCY REDUCTION TASK for parameters $\alpha(n)$, $\beta(n)$,  $s(n)$, $l(n)$ and $a=8$.

\end{theorem}
\begin{proof}
Suppose there exist two computable functions $f_1$ and $f_2$ satisfying the DEPENDENCY REDUCTION TASK for the given parameters and let $f(x,y) = E(f_1(x,y), f_2(x,y))$, where $E: \zo^{l(n)} \times \zo^{l(n)} \mapping \zo^{m_E(n)}$ is the Kolmogorov extractor from Theorem~\ref{t:extractor} for parameters $m_E(n) = \alpha(n)$, $s_E(n) = \beta(n)+ 8 \log n$ and dependency $\alpha_E(n) = \beta(n)$.
Theorem~\ref{t:twodependentsources} promises two strings $x$ and $y$ in $\zon$ such that $C(x \mid y) \geq s(n)$, $C(y \mid x) \geq s(n)$ and $C(f(x,y)) \leq m_E(n) - \alpha(n) + \log n + 2 \log \alpha (n) + O(1) = \log n + 2 \log \alpha(n) + O(1)$. Note that ${\rm dep}(x,y) \leq \alpha(n)$.

Let $u = f_1(x,y), v = f_2(x,y)$. The assumption implies that $C(u) \geq s_E(n)$, $C(v) \geq s_E(n)$ and ${\rm dep}(u,v) \leq \alpha_E(n)$.
 The extractor $E$ guarantees that  $C(E(u,v)) \geq m(n) - \alpha_E(n) - O(1) = \alpha(n) - (\alpha(n) - \log n - 3 \log \alpha(n)) - O(1) = 3 \log \alpha(n) + \log n - O(1)$. Since $E(u,v) = f(x,y)$, this is in conflict with the previous inequality.~\qed

\end{proof}

\if01
Our next theorem shows that no function $f$ can amplify independence from $\alpha$ to less than $\alpha - \log n - O(1)$.
\begin{theorem}
\label{t:noindepamplific}
Let $f: \zo^n \times \zo^n \mapping \zo^m$ be a uniform function in $n$ and let $\alpha \in \nat$, $\alpha \leq m$. For every $x \in \zo^n$ there exists $y \in \zo^n$ such that
\[
\begin{array}{ll}
n - C(y \mid x) & \leq \alpha  \\
m- C(f(x,y) \mid x) & \geq \alpha - \log n - O(1).
\end{array}
\]
\end{theorem}
\begin{proof}
Suppose first that $m = \alpha$. 
Fix $x \in \zo^n$. Take $z$ the most popular element in $A = \{f(x,y) \mid y \in \zo^n\}$.
Then $C(z \mid x) < \log n + O(1)$.
Also $|f^{-1}(z) \cap A| \geq 2^{n-\alpha}$. 
Therefore there exists $y$ with $(x,y) \in g^{-1}(z)$ with $C(y \mid x) \geq n-\alpha$.

If $m > \alpha$, let $g(x,y)$ be the prefix of length $\alpha$ of $f(x,y)$. Observe that $m- C(f(x,y) \mid x) \geq \alpha - C(g(x,y)\mid x) - O(1)$, because $C(f(x,y)\mid x) \leq C(g(x,y)\mid x) + (m-\alpha) + O(1)$. Since for every $x$, there exists $y$ with $C(y\mid x) \geq n- \alpha$ and  $C(g(x,y)\mid x) \leq \log n + O(1)$ (by the first part of the proof), the conclusion follows.~\qed

\end{proof}
\begin{theorem}
Let $f_1, f_2 : \zo^n \times \zo^n \mapping \zo^m$, two uniform functions in $n$ and let $\alpha \in \nat$, $\alpha \leq m$. Then there exists a pair of strings $x \in \zo^n, y \in \zo^n$ such that
\[
\begin{array}{ll}
n - C( y \mid x ) & \leq \alpha  \\
m- C(f_1(x,y) \mid f_2(x,y)) & \geq \alpha - \log n - O(1).
\end{array}
\] 

\end{theorem}
\begin{proof}
We define the functions $g_i$, $i=1,2$ as follows: $g_i(u,v)$ is the prefix of length $\alpha$ of $f_i (u,v)$.
Then for all $u, v$, $C(f_1(u,v) \mid f_2(u,v)) \leq C(g_1(u,v) \mid g_2(u,v)) + (m-\alpha) + O(1)$. By Proposition~\ref{p:onesourceimposs},  there exist two strings $x$ and $y$ such that $n - C(y \mid x) \leq \alpha$ and $C(g_1(x,y)) < \log n + O(1)$. Therefore
$C(g_1(x,y) \mid g_2(x, y)) < \log n + O(1)$, and the conclusion follows.~\qed 
\end{proof}
\fi
\if01
\section{Number of dependent strings}

Given a string $x \in \zo^n$, and $\alpha \in \nat$, how many strings have dependency with $x$ at least  $\alpha$? That is we are interested in estimating the size of the set
\[
A_{x, \alpha} = \{y \in \zo^n \mid C(y) - C(y \mid x) \geq \alpha \}.
\]
This is the set of strings about which, roughly speaking, $x$ has at least $\alpha$ bits of information.

The set $A_{x, \alpha}$ is included in $\{ y \in \zo^n \mid C(y \mid x) < n - \alpha + c \}$ for some constant $c$, and therefore $|A_{x, \alpha}| \leq C \cdot 2^{n-\alpha}$, for $C= 2^c$.

For a lower bound of $|A_{x, \alpha}|$, the interesting case is when $x$ does not have small Kolmogorov complexity (because a low-Kolmogorov complexity string cannot have information about too many strings).
We also give $m-t$. The table $T$ can be constructed from $n$ which is the length of $y$. 
\begin{proposition} (STACS construction -done more carefully)
Let $\alpha :\nat \mapping \nat$ and $\sigma : \nat \mapping \nat$ be two computable functions such that $\sigma(n) \geq 8 \log n$, for every $n \in \nat$. Let $m(n)$ be a computable function verifying $m(n) \leq \sigma(n)/2 - 4 \log n$. There exists a computable function $f:\zo^* \times \zo^* \mapping \zo^*$ of type $(n,n, m(n))$ such that for every sufficiently large $n$ and for every two strings $x$ and $y$ satisfying:
\begin{enumerate}
	\item $|x| = |y| = n$,
	\item $C(x) > \sigma$, $C(y) > \sigma(n)$, 
	\item $(x,y)$ have dependency at most $\alpha(n)$
\end{enumerate}
it holds that $C(f(x,y) \mid x) \geq m(n) - \alpha(n) - 4 \log n$ and $C(f(x,y) \mid y) \geq m(n) - \alpha(n) - 4 \log n$.
\end{proposition}
\begin{proof}
(Sketch) Take $s = \sigma/2$ and $m = \sigma/2 - 4 \log n$. Build a table $T$ that is $(S,M)$-balanced (according to def in the STACS paper). $s$ and $m$ satisfy the requirements for the construction of the table.

A bad column $u$ can be described by: color: $m$ bits; index among bad columns: $s$ bits; table $T$: $\log n$ bits.
So if $u$ is bad, then $C(u) \leq s+m+2\log s +\log n +2\log\log n +O(1) \leq s+m + 4 \log n = \sigma$.

So $y$ is a good column.

Then $x$, given $y$ can be described by: $C(z\mid y)$ bits; index $t_1-m +2$ bits (because index of $x$ is $\leq 2/2^m \cdot 2^{t_1+1} = 2^{t_1-m+2}$; table $T$: $\log n$ bits.

So, $C(x \mid y) \leq C(z \mid y) + (t_1 - m) + 2 \log(t_1-m) + \log n + 2 \log \log n +O(1) \leq C(z\mid y) + t_1 - m + 4 \log n$.

But $C(x\mid y) \geq t_1 - \alpha$. So, $t_1 - \alpha \leq C(z \mid y) + t_1 - m + 4 \log n$, and, therefore,
$C(z \mid y) \geq m - \alpha - 4 \log n$.~\qed

\end{proof}

\begin{theorem} Let $\alpha$ and $\sigma$ be natural constants satisfying $\alpha < n/2 - 8 \log n$, $n \geq \sigma > 2 \alpha  + 16 \log n$. Let $x \in \zo^n$ with $C(x) > \sigma$. Then $|A_{x, \alpha}| \geq (1/n^{6}) \cdot 2^{n-\alpha} - n^{16} 2^{2\alpha}$, provided $n$ is sufficiently large. Moreover,
$|\{ y \in \zo^n \mid C(y) \geq \sigma, C(y) - C(y \mid x) \geq \alpha \}| \geq (1/n^{6}) \cdot 2^{n-\alpha} - 2^{\sigma}$, provided $n$ is sufficiently large.

\end{theorem}
\begin{proof}
We prove the second assertion (the first one follows from the second one by taking $\sigma = 2\alpha + 16 \log n$).

Let $m = \alpha + 6 \log n$. By the previous proposition, if $C(y) \geq \sigma$ and $y \in \overline{A_\alpha}$, then it is possible from $x$ and $y$ to effectively construct a string $z$ of length $m$ with $C(z \mid x) \geq m - \alpha -  4 \log n = 2\log n$. 

Let $f: \zo^n \times \zo^n \mapping \zo^m$ be the extractor. 

Let $z$ be the most popular image for strings in the set $\{ (x,y) \mid y \in \zo^n\}$. Then
$C(z\mid x) < \log n + c$, for some constant $c$,  and $z$ has at least $2^{n-m}$ preimages in $\{ (x,y) \mid y \in \zo^n\}$. It follows that all these preimages are bad for extraction (because $m - \alpha - 4 \log n > \log n + c$, for $n$ sufficiently large). More precisely, let ${\rm BAD}_1 = \{y \in \zo^n \mid C(y) < \sigma\}$ and let ${\rm BAD}_2 = \{y \in \zo^n \mid C(y) \geq \sigma \mbox{ and }C(y \mid x) < C(y) - \alpha \}$. Then the second projection of $f^{-1}(z)$ is included in ${\rm BAD}_1 \cup {\rm BAD}_2$, and thus
\[
|{\rm BAD}_1| +  |{\rm BAD}_2| \geq 2^{n-m}.
\]
Since $|{\rm BAD}_1| < 2^\sigma$, it follows that 
\[
|{\rm BAD}_2| \geq 2^{n-m} - 2^\sigma \geq \frac{1}{n^{6}} 2^{n-\alpha} - 2^\sigma.
\]

\end{proof}

{\bf Remark:}  With the new extraction theorem, I do not obtain much better.
I obtain $cn^6$ instead of $n^{13}$, but this needs $\sigma = ((2k+1)/k) \alpha + O(\log n)$, so $\sigma > 2 \alpha + O(\log n)$. Also $n^{13}$ can probably be improved to something closer to $cn^6$.
\smallskip

Let $A_{x,\alpha} = \{u \in \zo^n \mid C(u \mid x) \leq C(u) - \alpha \}$.

For $u \in \zo^n$, we define $D_\alpha(u) = \{x \mid u \in A_{x,\alpha}\} = \{x \mid C(u) - C(u \mid x) \geq \alpha \}$
and $d_{\alpha}(u) = |D_{\alpha}(u)|$. Let $d_\alpha (u)$ be the number of sets $A_{x,\alpha}$ such that $u \in A_{x, \alpha}$.

The above theorem (with the alleged improvement) states that
\[
\frac{1}{n^6} 2^{n-\alpha} - (n^{16}) 2^{2\alpha} < |A_{x, \alpha}| <  2^{n-\alpha + c}.
\]
\begin{lemma}
If $x_1, x_2$ are at most $\beta$-dependent, then
\[
|A_{x_1, \alpha} \cap A_{x_2, \alpha}| \leq n^{20} 2^{n-2\alpha + \beta}.
\]
\end{lemma}
\begin{proof} We have
\[
C(x_1 x_2) \geq C(x_1) + C(x_2) - \beta - 3\log n.
\]
Let $u \in A_{x_1, \alpha} \cap A_{x_2, \alpha}$.  So, 
\[
C(u) - C(u \mid x_i) \geq \alpha, i = 1,2.
\]
Then, by Symmetry of Information, 
\[
C(x_i) - C(x_i \mid u) \geq \alpha - 5 \log n, i=1,2.
\]
Therefore, there exists programs $p_1$ and $p_2$ of length at most $C(x_1) - \alpha + 5 \log n $, and
respectively, $C(x_2) - \alpha + 5 \log n $ such that $U(p_1,u) = x_1$ and $U(p_2, u) = x_2$. This implies that, given $u$, $x_1 x_2$ can be constructed from $p_1$ and $p_2$. Therefore,
\[
\begin{array}{ll}
C(x_1 x_2 \mid u) & \leq |p_1| + |p_2| + 2 \log p_1 + O(1) \\
& \leq C(x_1) - \alpha + C(x_2) - \alpha + 2 \log |p_1| + 10 \log n  \\
& \leq C(x_1 x_2) - (2\alpha - \beta) + 15 \log n . 
\end{array}
\]
So, $C(xy) - C(xy \mid u) \geq 2 \alpha - \beta - 15 \log n$, which, by symmetry of information, implies
$C(u) - C(u \mid xy) \geq 2 \alpha - \beta - 20 \log n$.
Therefore $C(u \mid xy) \leq C(u) - 2\alpha + \beta + 17 \log n \leq n - 2\alpha + \beta + 20 \log n$.
Thus, $A_{x_1, \alpha} \cap A_{x_2, \alpha} \subseteq \{u \in \zo^n \mid C(u \mid xy) \leq n - 2\alpha + \beta + 20 \log n\}$. The conclusion follows.~\qed
\end{proof}

\begin{lemma} 
\label{l:degree}
For every $u \in \zo^n$, with $C(u) \geq 2\alpha +16 \log n$,
\[
\frac{1}{n^{11}} 2^{n-\alpha} - n^{21} 2^{2\alpha} \leq d_\alpha (u) \leq n^5 \cdot 2^{n-\alpha + c}.
\]
\end{lemma}
\begin{proof}
For every $x \in A_{u, \alpha + 5 \log n}$,
\[
C(x) - C(x \mid u) \geq \alpha + 5 \log n
\]
which by symmetry of information implies
\[
C(u) - C(u \mid x) \geq \alpha + 5 \log n - 5 \log n = \alpha,
\]
and therefore, $u \in A_{x, \alpha}$. Thus
\[
d_\alpha(u) \geq |A_{u,\alpha + 5 \log n}| \geq \frac{1}{n^6} 2^{n-\alpha - 5\log n}- n^{16} 2^{2(\alpha + 5 \log n)} = \frac{1}{n^{11}} 2^{n-\alpha}- n^{21} 2^{2\alpha}.
\]
For every $u \in \zo^n$,
\[
\begin{array}{ll}
u \in A_{x, \alpha} & \Rightarrow C(u) - C(u \mid x) \geq \alpha \\

& \Rightarrow C(x) - C(x \mid u) \geq \alpha - 5 \log n \\

& \Rightarrow C(x) \mid u) \leq n - \alpha + 5 \log n + c.
\end{array}
\]
Thus, $d_\alpha(u) \leq |\{x \in \zo^n \mid C(x \mid u) \leq n - \alpha + 5 \log n + c\}| \leq n^5 \cdot 2^{n-\alpha + c}$.~\qed
\end{proof}

\begin{theorem}
For every $\alpha < (1/3)(n- 32\log n -1)$, there exists a set $B \subseteq \zo^n$ of size $\poly(n) 2^{2\alpha}$ (more precisely the size of $B$ is bounded by $2n^{12}\cdot 2^{\alpha} + n^{16}2^{2\alpha}$) such that each string in $\zo^n$ is $\alpha$-dependent with some string in $B$. 
\end{theorem}
\begin{proof}
We choose $T= 2n^{12}2^\alpha$ strings $x_1, \ldots, x_T$, uniformly at random in $\zo^n$.
The probability that a fix $u$ with $C(u) \geq 2\alpha + 16 \log n$ does not belong to any of the sets $A_{x_i, \alpha}$, for $i \in [T]$, is at most
$(1 - \frac{1}{2n^{11} 2^\alpha})^T < e^{-n}$ (we took into account that 
$(1/n^{11}) 2^{n-\alpha} - n^{21} 2^{2\alpha} \geq (1/2n^{11}) 2^{n-\alpha}$. By the union bound, the probability that there exists
$u \in \zo^n$ that does not belong to any of the sets  $A_{x_i, \alpha}$, for $i \in [T]$, is bounded by $2^n \cdot e^{-n} < 1$. Therefore there are strings $x_1, \ldots, x_T$ in $\zo^n$ such that $\bigcup A_{x_i, \alpha}$ 
contains all the strings $u \in \zon$ having $C(u) \leq 2 \alpha + 16 \log n$.
The number of strings $u$ not yet covered is at most $n^{16}2^{2\alpha}$.
~\qed

PLAN TO ESTIMATE THE SIZE of a maximal set of independent strings: start with the set $B$ from above. Let $x$ be in $B$.
The number of independent strings in $A_{x, \alpha}$ should be around $2^{n/\alpha}$, because if $z_1, \ldots, z_T$ are such strings they each describe independently about $\alpha$ bits of $x$ (?). Actually we have $x$ is in the intersection of the sets $A_{z_i, \alpha}$, and such an intersection becomes empty if $T >> 2^{n/\alpha}$. 

\end{proof}
\fi



\if01
\appendix
\section{Appendix}
\medskip

{\bf Proof of Theorem~\ref{t:impossdistributions}.}

\begin{proof}
Suppose first that $m = \alpha$. Let $a$ be the most popular string in the image of $f$. Then $|f^{-1}(a)| \geq 2^{2n-m}$. Take (arbitrarily) $B \subseteq f^{-1}(a)$ with $|B| = 2^{2n-m}$. Consider
$\mbox{LEFT-B}$ the multiset of $n$-bit prefixes of strings in $B$ and $\mbox{RIGHT-B}$ the multiset of $n$-bit suffixes of strings in $B$. The multiplicity of a string $x$ in $\mbox{LEFT-B}$ is equal to the number of strings in $B$ that have $x$ as their left half. Thus each string in $\mbox{LEFT-B}$ has multiplicity at most $2^n$. Counting multiplicities $\mbox{LEFT-B}$ has $2^{2n-m}$ elements. Therefore $\mbox{LEFT-B}$ has at least $2^{n-m}$ distinct strings. The same holds for $\mbox{RIGHT-B}$. We take $X$ to be the random variable obtained by choosing uniformly at random one element in the multiset $\mbox{LEFT-B}$ and $Y$ is the random variable obtained by choosing uniformly at random one element in the multiset $\mbox{RIGHT-B}$. By the above discussion for each $x \in \zo^n$ and $y \in \zo^n$,
\[
\begin{array}{ll}
\prob[X=x] & \leq \frac{2^n}{2^{2n-m}} = \frac{1}{2^{n-m}}, \\
\prob[Y=y] & \leq \frac{2^n}{2^{2n-m}} = \frac{1}{2^{n-m}}, \\
\prob[X=x, Y=y] & \leq \frac{1}{2^{2n-m}}.
\end{array}
\]
Thus, $X$ and $Y$ satisfy the requirements, and $\prob[f(X,Y) = a] = 1$.

Suppose now that $m > \alpha$. We define $g, h :\zo^n \times \zo^n \mapping \zo^{\alpha}$ by
$g(x,y) = $ prefix of length $\alpha$ of $f(x,y)$ and $h(x,y) = $ suffix of length $m-\alpha$ of $f(x,y)$. Let $a \in \zo^{\alpha}$ and  the random variables $X$ and $Y$ defined as in the first part of the proof (\ie, the case $m=\alpha$) but with $g$ replacing $f$. Note that $\prob[g(X,Y) = a] = 1$. Let $b$ be a string in $\zo^{m-\alpha}$ such that $h^{-1}(b)$ has at least $2^{2n}/2^{m-\alpha}$ elements. Then
\[
\begin{array}{ll}
\prob [ f(X,Y) = ab] & = \prob[g(X,Y) = a, h(X,Y) = b] \\
&= \prob[h(X,Y) = b] \\
& \geq \frac{2^{2n}/2^{m-\alpha}}{2^{2n}} = 2^{-(m-\alpha)}.

\end{array}
\]
This concludes the proof.~\qed
\end{proof}
\medskip

{\bf Proof of Lemma~\ref{l:tables}}
\medskip

\begin{proof}
We use the probabilistic method. We show that a randomly colored table fails with probability  $< 1/2$ to satisfy the proper coloring property with respects to columns (property (a) in definition~\ref{d:tables}). A similar calculation shows the similar fact about proper coloring with respects to rows (property (b) in definition~\ref{d:tables}). Therefore we can conclude that a $(S,D)$-rainbow balanced table exists.

Observe that it is enough to consider sets $B_1$ of size exactly $S$ (because a set of size $kS$ can be broken into $k$ sets of size $S$ and if each smaller set satisfies the property, then the larger set will satisfy it as well).

Therefore, let us fix $B_1$ and $B_2$ subsets of $[N]$ of size $S$, let $B_1=\{u_1 < \ldots < u_S\}$ and  $B_1=\{v_1 < \ldots < v_S\}$. We fix $(A_1, \ldots, A_S) \in (\calA_D)^S$.

Let $X_{i,j}$ be the random variable which is $1$ if the cell$(u_i, v_j)$ is properly colored with respect to columns in $B_2$  and $(A_1, \ldots A_S)$ (\ie, $T(u_i, v_j) \in A_j$), and $0$ otherwise.
Then 
\[
\prob[X_{i,j} = 1] = \frac{|A_j|}{M} = \mu_j \in [1/D, m^2/D].
\]
Let $X = \sum_{i \in B_1, j \in B_2} X_{i,j}$. Then
\[
\mu = E[X] = \sum_{j \in B_2} \sum_{i \in B_1} E[X_{i,j}] = \sum_{j \in B_2}S \cdot \mu_j \in [S^2/D, S^2 \cdot m^2/D].
\]
By the Chernoff bounds,
\[
\prob[X \geq 2 \mu] \leq e^{-(1/3) \mu} \leq e^{-(1/3) (S^2/D)}.
\]
It follows that
\[
\prob[X \geq 2 \frac{S^2 m^2}{D}] \leq \prob[ X \geq 2 \mu] \leq e^{-(1/3) (S^2/D)}.
\]
We next take the union bound over all possible choices of $(A_1, \ldots, A_S) \in (\calA_D)^S$, and all possible choices of $B_1$ and $B_2$ subsets of $[N]$ of size $S$.

For $T \in [M/D, M\cdot m^2/D]$, the number of sets in $[M]$ of size $T$ is ${M \choose T} \leq (\frac{eM}{T})^T = e^T \cdot e^{T \ln (M/T)} \leq e^{T + T \ln D}$.
So the number of subsets of $[M]$ with sizes between $M/D$ and $M \cdot m^2/D$ is at most
\[
\sum_{T=M/D}^{M\cdot m^2/D} e^{T(1+ \ln D)}.
\]
Denoting $q = e^{(1+ \ln D)}$, the above sum is
\[
\begin{array}{ll}
\sum_{T=M/D}^{M\cdot m^2/D} e^{T(1+ \ln D)} & = q^{(M/D)} + q^{(M/D) + 1} + \ldots + q^{(M/D)m^2} \\
& = q^{(M/D)} \frac{q^{(M/D)(m^2 - 1) + 1} - 1}{q-1} \\
& < q^{(M/D)} \cdot q^{(M/D)m^2} \cdot q^{-(M/D)} \cdot \frac{q}{q-1} \\
& < 2 q^{(M/D)\cdot m^2} = 2 \cdot e^{(1 + \ln D) \cdot (M/D)\cdot m^2}.

\end{array}
\]
So the number of tuples $(A_1, \ldots, A_S) \in (\calA_D)^S$ is less than
$2^S \cdot e^{S \cdot (1 + \ln D) \cdot (M/D)\cdot m^2}$.

The number of ways of choosing $B_1$ and $B_2$ is
\[
{N \choose S} \cdot {N \choose S} \leq \big( \frac{eN}{S}\big)^{2S} = e^{2S + 2S \ln (N/S)}.
\]
For the union bound to give a probability $\leq e^{-1} < 1/2$ we need
\[
(1/3)(1/D)S^2 \geq S + S(1 + \ln D)(M/D)m^2 + 2S + 2S \ln(N/S) + 1,
\]
which holds true if the parameters satisfy the hypothesis.~\qed
\end{proof}
\fi
\end{document}